\documentclass[12pt]{article}
\usepackage[latin1]{inputenc}
\usepackage{natbib}

\usepackage[a4paper, hdivide={2.5cm,,2.5cm}, vdivide={3cm,,2cm}]{geometry}
\usepackage{amsmath}
\usepackage{amsfonts,amsgen,amsopn}
\usepackage{amssymb}

\usepackage{setspace}
\usepackage{color}

\usepackage{graphicx}
\usepackage[figuresright]{rotating}
\usepackage{enumerate}
\usepackage{fancyhdr}
\usepackage{pst-all}
\usepackage{multido}
\usepackage{pst-blur}
\usepackage{amsthm}
\usepackage{framed}
\usepackage{authblk}
\usepackage{xspace}

\newtheorem{theorem}{Theorem}
\newtheorem{lemma}{Lemma}

\theoremstyle{remark}

\newcommand{\pp}{\mathbb{P}}

\newcommand{\fitch}{\texttt{MP}}

\begin{document}
\bibliographystyle{acm}

\begin{titlepage}

\title{\large Maximum Parsimony on Subsets of Taxa}

\author{\normalsize Mareike Fischer} \affil{Biomathematics Research
  Centre, University of Canterbury, Private Bag 4800, Christchurch, New
  Zealand; E-mail: email@mareikefischer.de}

\setcounter{Maxaffil}{1}

\author{\normalsize Bhalchandra D. Thatte} \affil{
  Department of Statistics, University of Oxford, 1 South Parks Road,
  Oxford, OX1 3TG, United Kingdom; E-mail: thatte@stats.ox.ac.uk}

\date{}

\maketitle

{\footnotesize
{\bf Corresponding author: } Mareike Fischer\\


{\bf Running head: } MAXIMUM PARSIMONY ON SUBSETS OF TAXA}
\end{titlepage}

\newpage

\setcounter{page}{2}

\subsubsection*{\sc{Abstract}}
In this paper we investigate mathematical questions concerning the reliability (reconstruction accuracy) of Fitch's maximum parsimony algorithm for reconstructing the ancestral state given a phylogenetic tree and a character. In particular, we consider the question whether the maximum parsimony method applied to a subset of taxa can reconstruct the ancestral state of the root more accurately than when applied to all taxa, and we give an example showing that this indeed is possible. A surprising feature of our example is that ignoring a taxon closer to the root improves the reliability of the method. On the other hand, in the case of the two-state symmetric substitution model, we answer affirmatively a conjecture of Li, Steel and Zhang which states that under a molecular clock the probability that the state at a single taxon is a correct guess of the ancestral state is a lower bound on the reconstruction accuracy of Fitch's method applied to all taxa.
\newpage

\par
\vspace{0.5cm}

\label{sec-intro}
In a recent study \citep{li-steel-zhang-2008}, a likelihood analysis of
Fitch's maximum parsimony method \citep{fitch-1971} (which we call MP
for short) for the reconstruction of the ancestral state at the root was
conducted.  It was shown that in a rooted phylogenetic tree if a leaf
(taxon) is closer to the root than all the other leaves, then the
character state at this leaf may sometimes be a more accurate guess of
the ancestral state than the ancestral state constructed by MP applied
to all taxa. The authors also provided an example of a phylogenetic tree
for which MP for the reconstruction of the root state works more
reliably on a subset of taxa closer to the root than on all taxa.

Generally the root state is more likely to be conserved on taxa that are
nearer to the root than on taxa that are farther away. Therefore, it is
not surprising that on some trees the root state can be more reliably
estimated by looking at only taxa nearer to the root.  But can the
reconstruction accuracy of MP improve when a taxon or a subset of taxa
close to the root is ignored?  We
presented a surprising example of a tree on which MP on a subset of taxa
is more likely to reconstruct the correct ancestral state. In our
example, the reconstruction accuracy improves when we ignore a taxon
close to the root from our analysis. Moreover, the ignored taxon may be
arbitrarily close to the root compared to the taxa that are not
ignored. On the other hand, we show that under a molecular clock, considering a single taxon is never better than
considering all taxa for the purpose of ancestral state
reconstruction. Our analysis partially resolves a conjecture of Li,
Steel and Zhang. They conjectured that under a molecular clock, maximum
parsimony on all taxa is expected to generally perform at least as good
(in the sense of the reconstruction accuracy) as reconstructing the
ancestral state based on the character state at a single taxon. 
We make the conjecture precise and answer it
affirmatively for the case of the two-state symmetric model.

\section*{\sc{Maximum Parsimony on Subsets of Taxa}}
\label{sec-subset}
We start with some notation. Let $T$ be a rooted binary phylogenetic
tree on the leaf set (i.e., the set of taxa) $X$. Let the root of the
tree be $\rho$. We assume that each vertex in $T$ takes one of the two
states $\alpha $ and $\beta$.  The states evolve from the root state
under a simple symmetric model of substitution described as
follows. Suppose that $e=(u,v)$ is an edge of the tree $T$, where $u$ is
the vertex closer to the root than $v$ is. Let $p_e$ be the {\em
  substitution probability} on edge $e$: it is the probability that $v$
is in state $\beta$ conditional on $u$ being in state $\alpha$, and is
denoted by $\pp(v = \beta|u=\alpha)$. The model is assumed to be
symmetric, therefore, $p_e = \pp(v = \beta|u=\alpha) = \pp(v =
\alpha|u=\beta)$. Moreover we assume that $p_e \leq 1/2$. This model is
also known as the two-state Neyman model, and has been discussed
elsewhere in the literature \citep[for example,][]{cavender-1978,
  farris-1973}, and is a special case of the well known symmetric
$r$-state model \citep[see, e.g.,][]{tuffley-steel-1997}. A {\em binary
  character} is an assignment of one of the two possible states $\alpha$
and $\beta$ to each leaf of the tree, that is, it is a map $f: X \to
\{\alpha, \beta \}$.

In this section we analyze the probability that maximum parsimony
applied to a subset of the set of taxa correctly estimates the true
state at the root. Suppose that $Y$ is a subset of $X$ (denoted
$Y\subseteq X$). It induces a subtree $T_Y$, rooted at a vertex
$y$. Here $y$ is the most recent common ancestor of vertices in $Y$. It
is possible that $y = \rho$. Let $f_Y$ denote the restriction of a
binary character $f$ to $Y$. MP assigns states $\alpha$ or $\beta $ to
all internal nodes (including the root $\rho $) so that the total number
of substitutions is minimized. Such an assignment is not necessarily
unique: MP computes a set $S_z$ of possible states at each internal
vertex $z$, so that each most parsimonious assignment must assign one of
the states in $S_z$ to the vertex $z$. When MP is applied to a binary
character $f$, we have either $S_{\rho} = \{\alpha \}$ or $S_{\rho} =
\{\beta\}$ or $S_{\rho} = \{\alpha,\beta\}$ at the root $\rho$. If
$S_{\rho}$ is either $\{\alpha \}$ or $\{\beta\}$, then we say that MP
{\em unambiguously reconstructs} the root state; otherwise (when
$S_{\rho}$ is $\{\alpha,\beta\}$) we say that MP {\em ambiguously
  reconstructs} the root state.

The maximum parsimony algorithm may also be applied to $f_Y$ on the
subtree $T_Y$. It returns a state set $S_y= \{\alpha \}$ or $S_y=\{\beta\}$ or
$S_y=\{\alpha,\beta\}$ for the root $y$ of $T_Y$.

In the following, we will denote by $\fitch(f,T)$ the set of character states chosen by Fitch's maximum parsimony algorithm as possible root states when applied to a character $f$ on a tree $T$.

Li, Steel and Zhang defined the {\em unambiguous reconstruction
  accuracy} $UA(Y)$ and the {\em ambiguous reconstruction accuracy}
$AA(Y)$ as follows:

\begin{eqnarray*}
  UA(Y) & := & 
  \pp\left(\fitch(f_Y, T_Y) = \{\alpha\}|\rho = \alpha\right), \\[3mm]
  AA(Y) & := &
  \pp\left(\fitch(f_Y, T_Y)= \{\alpha, \beta\}|\rho = \alpha\right).
\end{eqnarray*}

In other words, $UA(Y)$ is the probability that the root state $\alpha$
evolves to a character $f$ for which maximum parsimony on $Y$ assigns
the state set $\{\alpha \}$ to the root $y$ of $T_Y$.

Furthermore, they defined the {\em reconstruction accuracy} as
\begin{equation}
  RA(Y) = UA(Y) + \frac{1}{2}AA(Y),
\end{equation}
where the second term indicates that when MP reconstructs the state at the
root ambiguously, we select one of the states with equal probability.

Note that MP, when applied to $Y$, estimates a state at the root vertex
$y$ of the subtree $T_Y$ induced by $Y$. Since it is possible that the
root $y$ of $T_Y$ is different from the root $\rho$ of $T$, we define
the reconstructed state at $y$ to be the estimate of the state at the
root based on the subset $Y$ of taxa.

Li, Steel and Zhang gave an example of a tree for which the
reconstruction accuracy of MP on a proper subset of taxa is higher than
the reconstruction accuracy of MP on all taxa, i.e., $RA(Y) > RA(X)$ for
some proper subset $Y$ of $X$. But their example requires that the taxa
in $Y$ are closer to the root than the taxa not in $Y$, i.e., that the probability of a substitution from the root to any taxon in $Y$ is smaller than the probability of a substitution from the root to the other taxa. The example that
we present in the following subsection does not require any taxa to be
closer to the root. On the contrary, our example shows that a {\em
  misleading taxon or taxa} (a taxon or taxa that have an adverse effect
on the reconstruction accuracy) may be arbitrarily close to the root.

\subsection*{\sc{An Example of a Misleading Taxon}}
\label{sec-misleading}
The main result of this section is the following theorem which shows
that there are trees on which the reconstruction accuracy improves when
a taxon close to the root is ignored in an MP based ancestral state
reconstruction. Moreover, such a misleading taxon may be arbitrarily
close to the root.

\begin{theorem}
\label{thm-misleading}
Let $p_z$ be any real number such that $0 < p_z < 1/2$. Then there
exists a binary phylogenetic tree $T$ on a leaf set $X$ and rooted at
$\rho$ such that the following conditions are satisfied:
\begin{enumerate}
\item for some leaf $z$, the substitution probability from $\rho$ to $z$
  is $p_z$;
\item $RA(X \backslash \{z\})> RA(X)$; and 
\item for each leaf $v \neq z$, the substitution probability $p_v$ from
  $\rho$ to $v$ is more than $p_z$, i.e., $z$ is closer to the root than
  any other taxon.
\end{enumerate}
\end{theorem}

To prove the above theorem, we first need some notation and a lemma. Let
$y$ be a vertex in a binary phylogenetic tree $T$, and let $Y$ be the
set of leaves below $y$. We associate three probabilities with $Y$ as
follows.
\begin{eqnarray*}
  P_{\alpha}(Y) & := & 
  \pp\left(\fitch(f_Y, T_Y) = \{\alpha\}|y = \alpha\right), \\[3mm]
  P_{\beta}(Y) & := & 
  \pp\left(\fitch(f_Y, T_Y) = \{\beta\}|y = \alpha\right), \\[3mm]
  P_{\alpha \beta}(Y) & := &
  \pp\left(\fitch(f_Y, T_Y)= \{\alpha, \beta\}|y = \alpha\right).
\end{eqnarray*}

Let $T_n$ be a balanced binary tree of depth $n$, i.e., with $n$ edges on the path from the root to each leaf. Let $X$ be its leaf
set. Suppose that the substitution probability on each edge of $T_n$
is $q$. For this particular symmetric tree, we denote $P_{\alpha}(X)$,
$P_{\beta}(X)$ and $P_{\alpha \beta}(X)$ by $P_{\alpha}(n,q)$,
$P_{\beta}(n,q)$ and $P_{\alpha \beta}(n,q)$, respectively. The
convergence properties of these probabilities (for $n \rightarrow
\infty$ and for various values of $q$) have been studied in detail in
\citep[][]{steel-charleston-1995} and \citep[][]{yang-2008}. We state their
result on the convergence of $P_{\alpha}(n,q)$ that additionally provides 
a lower bound on $P_{\alpha}(n,q)$ which is independent of $n$.

\begin{lemma}
  \label{lem-yang}\citep{steel-charleston-1995,yang-2008}
  Let $T_n$ be a binary balanced phylogenetic tree of depth $n \geq
  2$. Let $q < 1/8$ be the probability of substitution on each edge of
  the tree. Then $P_{\alpha}(n,q)$ approaches 
\[\frac{1}{2} \left(1- \frac{2q}{1-2q} + \frac{\sqrt{(1-8q)(1-4q)}}{(1-2q)^2}\right)\] from above as 
$n \rightarrow \infty $. Moreover, as $q$ goes to 0, the above limiting value approaches 1.
\end{lemma}

\begin{proof}[Proof of Theorem~\ref{thm-misleading}]
  Let $T$ be a phylogenetic tree rooted at $\rho $ constructed as
  follows. The left subtree of $T$ contains a single leaf $z$. The right
  subtree of $T$ is $T_Y$ with leaf set $Y$ and root $y$. Therefore, the leaf set of $T$ is $X = Y \cup \{z\}$. We choose
  $T_Y$ to be a balanced binary tree of depth $n$ and a substitution
  probability $q$ on each edge. Let the substitution probabilities on
  $(\rho, z)$ and $(\rho, y)$ be $p_z$ and $p_y$, respectively, where
  $p_z$ is any given real number such that $0 < p_z < 1/2$. (See an
  illustration of these parameters in Figure 1.

  For the above tree, the reconstruction accuracy on $X$ is given by
  \begin{eqnarray*}
    RA(X) & = & 
    (1-p_z)\left((1-p_y)P_{\alpha}(n,q) + p_yP_{\beta}(n,q) 
      + P_{\alpha \beta}(n,q)\right) \\[1mm]
    && + \frac{1}{2}\,
    p_z\left((1-p_y)P_{\alpha}(n,q) + p_yP_{\beta}(n,q)\right)
    \nonumber \\[1mm]
    && + \frac{1}{2}\,
    (1-p_z)\left(p_yP_{\alpha}(n,q)+ (1-p_y)P_{\beta}(n,q)\right).
  \end{eqnarray*}
  
  The reconstruction accuracy on $Y$ is given by
  \begin{equation*}
    RA(Y) = (1-p_y)P_{\alpha}(n,q) 
    + p_yP_{\beta}(n,q) 
    + \frac{1}{2}\,P_{\alpha \beta}(n,q).
  \end{equation*}

  In order to satisfy $RA(Y) > RA(X)$, we therefore must have
  \begin{equation}
    \label{eq-constr1}
    (p_z-p_y)P_{\alpha}(n,q) > (1-2p_z)P_{\alpha \beta}(n,q) 
    + (1-p_z-p_y)P_{\beta}(n,q).
  \end{equation}

  We now show that for any value of $p_z$ however small, the remaining
  substitution probabilities $q$ and $p_y$ and the depth $n$ of $T_Y$
  can be chosen such that $RA(Y)> RA(X)$ (condition 2 in
  Theorem~\ref{thm-misleading}), and for every vertex $v$ in $Y$, the
  probability of a change of state from the root to $v$ is more than
  $p_z$ (condition 3 in Theorem~\ref{thm-misleading}).

  We express the third condition in Theorem~\ref{thm-misleading} in a
  different form. Let $Q:= 1-2q$, $P_z := 1-2p_z$ and $P_y := 1-2p_y$.
  Since the tree $T_n$ is symmetric, the probability of a change of
  state from the root to any leaf $v$ in $Y$ is the same, and is given
  by $p_v = \frac{1-P_yQ^n}{2}$.  Therefore, the third condition may now
  be written as $ P_yQ^n < P_z$, or equivalently as
  \begin{equation}
    \label{eq-constr3}
    (1-2q)^n < \frac{1-2p_z}{1-2p_y}.
  \end{equation}

  It follows from Lemma~\ref{lem-yang} that, for all $n\geq2$, as $q$ approaches $0$, the
  left hand side of Equation~(\ref{eq-constr1}) approaches $p_z - p_y$
  and the right hand side approaches 0. Therefore, there is a real
  number $\epsilon$ such that $0 < \epsilon < 1/8$, and whenever $q <
  \epsilon$, Equation~(\ref{eq-constr1}) is satisfied. Now given a value
  of $p_z$, we first arbitrarily fix $p_y$ such that $0 < p_y < p_z$,
  and then fix a value of $H:= (1-2q)^n$ satisfying the constraint in
  Equation~(\ref{eq-constr3}). We then choose $n$ sufficiently large so
  that $q = (1 - H^{1/n})/2 < \epsilon$ and the constraint given in
  Equation~(\ref{eq-constr1}) is satisfied as well. This completes the proof.
\end{proof}

Note that when $q \geq \frac{1}{8}$,
the sequence $P_{\alpha}(n,q)$ has quite
different convergence properties than when $q < \frac{1}{8}$, and the bound provided by Lemma~\ref{lem-yang} does not apply,
\citep[see ][for details]{steel-charleston-1995, yang-2008}. Therefore, our construction of a misleading taxon given in the proof of Theorem \ref{thm-misleading} strongly depends on $q$ being sufficiently small.

\subsection*{\sc{A Single Taxon Under a Molecular Clock}}
\label{sec-1taxon}
In this section, we consider binary characters on a binary phylogenetic
tree $T$ with leaf set $X$ under a molecular clock and the two-state
symmetric model introduced earlier. Let $p$ be the probability that a
leaf is in a different state than the root. Therefore, if we were to
guess the root state by looking at only one taxon, the probability of
success would be the probability that the root state was conserved at
this taxon, which is $1-p$. That is, if $Y = \{x_1\} $ is a single taxon
subset of $X$, then $RA(Y) = 1-p$. In the following, we show that $1-p$
is in fact a lower bound on $RA(X)$, implying that MP applied to all
taxa reconstructs the root state at least as successfully as
reconstructing the root state from a single taxon.

As shown in Figure 2, 
we denote the children of $\rho$ by
$y_1$ and $y_2$, and define $T_i$ to be the subtrees rooted at $y_i$ for
$i$ in $\{1,2\}$. Let the probabilities of a change of state from $\rho$
to $y_i$ be $p_i$. The probabilities of a change of state from $y_i$ to
any leaf under $y_i$ are $p^\prime_i$. For $i$ in $\{1,2\}$, we define
$P_i:=1-2p_i$. Similarly we define $P:=1-2p$.

In the above notation, we prove the following lower bound on $RA(X)$.

\begin{theorem}
  \label{thm-low} For any rooted binary phylogenetic ultrametric
  (clock-like) tree $T$ with leaf set $X$, the reconstruction accuracy
  of MP is at least equal to the conservation probability from the root
  to any leaf, that is,
  \[
  RA(X) \geq 1-p.
  \]
\end{theorem}

\begin{proof}
  We first state two recursions, which we then use to give an
  inductive proof of the theorem.

\begin{eqnarray*}
  P_{\alpha}(X) &=& 
  \left(\frac{1+P_1}{2}P_{\alpha}(Y_1) 
    + \frac{1-P_1}{2}P_{\beta}(Y_1)\right)
  \left(\frac{1+P_2}{2}P_{\alpha}(Y_2) 
    + \frac{1-P_2}{2}P_{\beta}(Y_2)\right) \\[2mm]
  && + P_{\alpha \beta}(Y_1)
  \left(\frac{1+P_2}{2}P_{\alpha}(Y_2) 
    + \frac{1-P_2}{2}P_{\beta}(Y_2)\right)  \\[2mm]
  && + \left(\frac{1+P_1}{2}P_{\alpha}(Y_1) 
    + \frac{1-P_1}{2}P_{\beta}(Y_1)\right)P_{\alpha \beta}(Y_2)
\end{eqnarray*}

\begin{eqnarray*}
  P_{\alpha \beta}(X) &=& 
  \left(\frac{1+P_1}{2}P_{\alpha}(Y_1) 
    + \frac{1-P_1}{2}P_{\beta}(Y_1)\right)
  \left(\frac{1-P_2}{2}P_{\alpha}(Y_2) 
    + \frac{1+P_2}{2}P_{\beta}(Y_2)\right) \\[2mm]
  && + \left(\frac{1-P_1}{2}P_{\alpha}(Y_1) 
    + \frac{1+P_1}{2}P_{\beta}(Y_1)\right)
  \left(\frac{1+P_2}{2}P_{\alpha}(Y_2) 
    + \frac{1-P_2}{2}P_{\beta}(Y_2)\right)  \\[2mm]
  && + P_{\alpha \beta}(Y_1)P_{\alpha \beta}(Y_2)
\end{eqnarray*}

We define $D(X) := P_{\alpha}(X) + P_{\alpha \beta}(X)/2
- (1+P)/2$, and similarly we define $D_1 := D(Y_1)$ and $D_2 :=
D(Y_2)$. The above recursions can be manipulated with a computer algebra
system to verify that
\begin{equation*}
  4 D(X) = 2 P_{\alpha \beta}(Y_1) D_2 P_2 
  + 2 P_{\alpha \beta}(Y_2) D_1 P_1 + 2 D_2 P_2 + 
  2 D_1 P_1 + P_{\alpha \beta}(Y_1) P 
  + P_{\alpha \beta}(Y_2) P. 
\end{equation*}

Now, by induction on the number of leaves, we show that $D(X)$ is
non-negative. The base case of the inductive proof is when $Y_1$ and
$Y_2$ are singleton sets, in which case $D(Y_1)$, $D(Y_2)$ and $D(X)$
are all equal to $0$, that is $RA(X)$ is $1-p$. Suppose that the tree
$T$ has $n$ taxa, and suppose that $D(X)$ is non-negative for all trees
having fewer than $n$ taxa. Since both $Y_1$ and $Y_2$ contain fewer
than $n$ taxa, $D(Y_1)$ and $D(Y_2)$ are both non-negative. Since
$P_{\alpha \beta}(Y_1)$, $P_{\alpha \beta}(Y_2)$, $P_1$
and $P_2$ are all non-negative, the right hand side of the above
equation is non-negative, implying the theorem.
\end{proof}

\section*{\sc{Discussion}}
In this paper we analyzed the question of how MP performs when used to
reconstruct the ancestral root state. In particular, we considered the
problem for phylogenetic trees on which the probability of a change of
state from the root vertex to any leaf is constant. Earlier simulation
studies \citep[e.g.,][]{salisbury-kim-2001, zhang-nei-1997} suggested
that the reconstruction accuracy is generally increased when more taxa
are considered. But simulations conducted by Li, Steel and Zhang showed
that even under a molecular clock, MP may perform better on certain
subsets of taxa. We present an example of
a tree in which one of the subtrees at the root consists of a single
leaf and a pending edge, and the other subtree is a balanced binary tree
of large depth and small ($< 1/8$) substitution probabilities on all
edges. On this tree, we observed that the ancestral state reconstruction
is more accurate if only the set of taxa on the balanced subtree is
considered. This is in contrast to the example given by Li, Steel and
Zhang in which an outgroup taxon closer to the root or a single fossil
record may give a better estimate of the root state than considering the
whole tree. As our example shows, even a very short edge connecting the
root with a leaf cannot guarantee an accurate root state estimation if
the remaining taxa induce a balanced tree with a large number of
taxa. For such trees, it may be better to ignore the fossil or a taxon
closer to the root.  Thus there seems to be no general theoretical
guideline to decide what subsets of taxa are to be used for a more
reliable reconstruction of the root state. In general we believe that
very long leaf edges would have an adverse effect on the ancestral state
reconstruction using MP, but it would be useful to quantitatively or
algorithmically state and prove such an expectation.

While using the data on a subset of taxa may give a more accurate
estimate of the root state, in general a single taxon subset does not
give a better reconstruction accuracy. We showed this
by resolving a conjecture of Li, Steel and
Zhang. They conjectured that for two state characters on an ultrametric
(clock-like) tree and a symmetric model of substitution, ancestral state
reconstruction using all taxa is at least as accurate as that using a
single taxon. We expect such a result to be true even when there are
more than two states.

\subsection* {\sc{ Acknowledgements} }
We take this opportunity to thank the Allan Wilson Centre for Molecular
Ecology and Evolution, New Zealand, for support, and Mike Steel for
introducing us to the problems discussed here and for many useful
conversations.

\newpage

\clearpage
\section*{\sc{References}}
 
\renewcommand*{\refname}{}

\newpage
\subsection*{Figure captions:}

\begin{enumerate}
\item[{\bf Figure 1:}] A tree on which MP is more accurate when applied
  to $Y\subset X$.
\item[{\bf Figure 2:}] Illustration for Theorem~\ref{thm-low}: For any
  clocklike binary phylogenetic tree T the reconstruction accuracy of MP
  based on all leaves is at least as good as the one based on a single
  leaf.
\end{enumerate}

\newpage

\subsection*{Figure 1}

\begin{figure}[ht]
\begin{center}
  \includegraphics[scale=1]{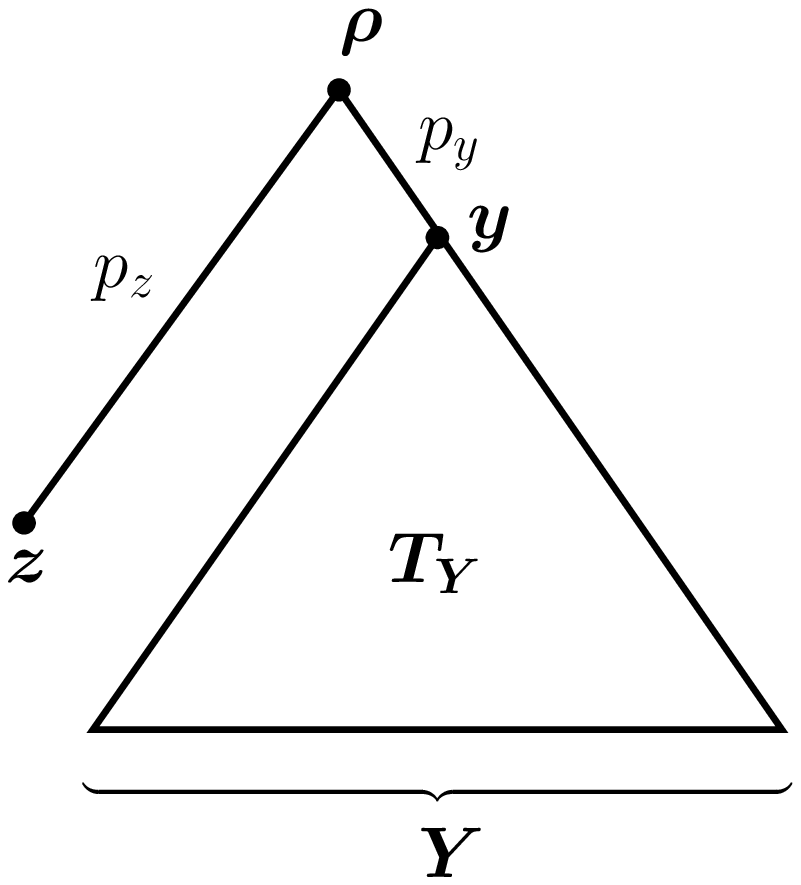}
\end{center}
\label{fig-rec-accuracy}
\end{figure}

\newpage
\subsection*{Figure 2}

\begin{figure}[ht]
\begin{center}
  \includegraphics[scale=1]{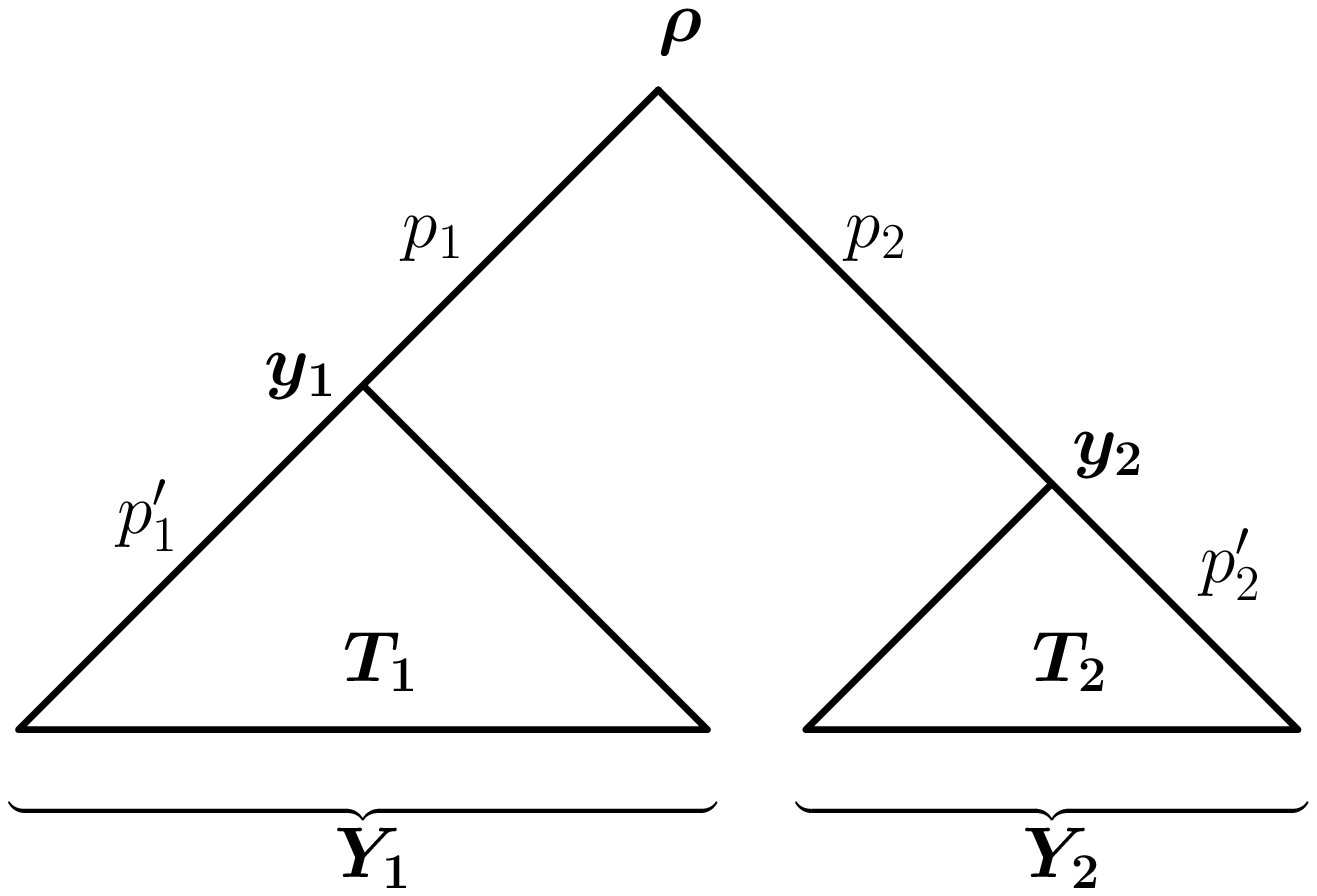}
\end{center}

\label{fig-recur}
\end{figure}

\end{document}